\documentclass{svproc}

\usepackage{url}

\usepackage{xcolor,amsmath,amssymb,bbm,enumerate}
\usepackage{hyperref}
\hypersetup{colorlinks=true}

\newcommand{\mathd}{\mathrm{d}}

\newcommand{\tmem}[1]{{\em #1\/}}
\newcommand{\tmop}[1]{\ensuremath{\operatorname{#1}}}

\newenvironment{enumeratenumeric}{\begin{enumerate}[1.] }{\end{enumerate}}

\begin{document}
\mainmatter

\title{ RepLAB: a computational/numerical approach to representation theory}

\author{Denis Rosset\inst{1} \and Felipe Montealegre-Mora\inst{2} \and Jean-Daniel Bancal\inst{3}}

\authorrunning{Denis Rosset et al.}

\institute{
  Perimeter Institute for Theoretical Physics, Waterloo, Ontario, Canada, N2L 2Y5
  \and
  Institute for Theoretical Physics, Universit{\"a}t zu K{\"o}ln, Cologne 50937 Germany
  \and
  D{\'e}partement de Physique Appliqu{\'e}e, Universit{\'e} de Gen{\`e}ve, 1211 Gen{\`e}ve, Suisse
}

\maketitle

\begin{abstract}
  We present a MATLAB/Octave toolbox to decompose finite dimensionial representations of compact groups.
  Surprisingly, little information about the group and the representation is needed to perform that task.
  We discuss applications to semidefinite programming.
\end{abstract}

Early in the development of quantum formalism, some regarded group theory as a mere nuisance: the label {\tmem{Gruppenpest}} is attributed to Pauli and his natural talent for derision~{\cite{Szanton1992}}.
While the usefulness of group theory in quantum physics is no longer a matter of debate, most uses involve Schur-Weyl duality and a handful of well-understood groups: symmetric and cyclic groups, the Pauli and Clifford groups, and the (special) unitary group.

It turns out that a large variety of groups are present in quantum information computations.
Thousands of families of Bell inequalities~{\cite{Rosset2014a}} have been discovered through the enumeration of local polytope facets~{\cite{Brunner2014}}: most of them have some degree of symmetry~{\cite{Sliwa2003}}.
In the context of communication games, symmetries can encode structural information; for example that only the sum of outputs modulo $d$ matters~{\cite{Collins2002}}.

Even when the underlying group is well understood, its representations can still be difficult to decompose.
For example, a body of literature is dedicated to the decomposition of the partially transposed tensor representations of the unitary group~{\cite{Mozrzymas2014b,Studzinski2013,Mozrzymas2018,Mozrzymas2014a}} corresponding to universal quantum cloning machines~{\cite{Studzinski2014}}.
Systems composed of multiple subsystems will have symmetries corresponding to the composition of basic groups.
Bell scenarios have a symmetry group consisting of relabelings of parties, inputs and outputs, most easily described as a double wreath
product of symmetric groups~{\cite{Renou2017}}.

In the last twenty years, quantum information has seen the rise of numerical methods, especially those based on semidefinite programming (SDP).
Density matrices correspond naturally to semidefinite program variables, and so do a variety of quantum objects~\cite{Rosset2019losr} through the use of the Choi isomorphism~{\cite{Choi1975,Jamiolkowski1972}}: for example, the SDP hierarchies for the separability problem~{\cite{Doherty2002,Doherty2004,Doherty2005,Navascues2009}}.
Approaches based on moments of noncommutative polynomials are widely used to compute a variety of bounds~{\cite{Navascues2008a,Navascues2007,Navascues2012,Navascues2015a,Wolfe2019}} using SDP relaxations.

In other fields, the processing of group representations is currently being automated.
In conformal bootstrap~{\cite{Poland2019}}, a package constructs symmetry-adapted SDPs~{\cite{Go2019}}, using key primitives identified by the authors\footnote{For conformal bootstrap, those primitives are: the list of irreducible representations of the group, explicit representations of the group generators using unitary matrices, the complex conjugation map and the decomposition of tensor products of irreducible representations.}.
We build a similar high-level software to perform symmetry reduction, but without dependencies on other computer algebra systems~\footnote{Our approach is numerical Alternatively, we could use the group character table~\cite{Schneider1990,Conlon1990,Baum1994}. Such methods provide exact answers, but sometimes fail to work in a reasonable amount of time on small groups~\cite{Go2019}. Our software weights currently 8000 lines of mostly standalone code, to be compared with the gigabyte required by full fledged CAS systems such as GAP~\cite{Gap}.}.
In the present work, we will identify the primitives necessary to speed up common computations in quantum information.
We aim to work with as little information as possible from the group; in exchange, we are satisfied with double precision floating point answers.
RepLAB derives from specialized code written by the first author for a specific quantum class of problems~{\cite{Tavakoli2019}}, and has since been rewritten from the ground-up to be used in a variety of contexts.
The latest version is available at \href{https://github.com/replab/replab}{https://github.com/replab/replab}.

This work is structured as follows.
In Section~\ref{Sec:Primitives}, we examine our motivations and identify the key primitive to speed up solving SDPs.
In Section~\ref{Sec:Algorithms}, we implement this key primitive using easier subtasks: sampling from the commutant of a representation and eigendecomposition.
In Section~\ref{Sec:Features}, we summarize other features of RepLAB; open and ongoing questions are discussed in Section~\ref{Sec:Open}.

\section{Motivations}
\label{Sec:Primitives}

A (complex\footnote{Note that quantum information problems often reduce to real semidefinite programs for reasons outside the scope of this paper; and
that the best SDP solvers support only real SDPs~{\cite{Gilbert2017a}}. RepLAB works both with real and complex representations.}) semidefinite program is
described by one of these forms~{\cite{Sturm2002}}.
\begin{equation}
  \begin{array}{cl|lc}
    \text{Primal form\ \ \ } &  &  & \text{Dual form\ \ \ \ }\\
& & & \\
    \begin{split}
     \min_{X = X^{\dag} \in \mathbbm{C}^{n \times n}} & \tmop{tr} \langle C, X \rangle \\
     \text{s.t.\ \ \ \ \ } &  X \succeq 0 \\
     &  \langle A_i, X \rangle = b_i, \ \ \forall\ i
    \end{split}
\ \ \ \ & & &\ \ \ \ \ \ \ \ 
    \begin{split}
    \max_{\vec{y} \in \mathbbm{R}^m}\ & \vec{b}^{\top} \cdot \vec{y} \\
     \text{s.t.\ } &  \chi = C - \sum_i y_i A_i  \succeq 0_{},
    \end{split}
  \end{array}
\end{equation}
where we denote SDP constraint by $\succeq 0$ and the conjugate transpose by $^{\dag}$.
The problem is specified by the Hermitian matrices $C$ and $\{ A_i \}$, and the real vector $\vec{b} \in \mathbbm{R}^m$.
In what follows, we assume that the canonical primal or dual SDP has already been prepared in a form that respects symmetry~\footnote{For example, the NPA hierarchy~{\cite{Navascues2008a,Navascues2007}} is naturally implemented in the dual form above; for pointers towards invariant constructions, see~{\cite{Rosset2018,Tavakoli2019}}.}.

\subsection{Invariant semidefinite programs}

An invariant program in the primal form has a matrix $X$ invariant~{\cite{Gatermann2004}} under a group representation (the same idea applies to $\chi$ in the dual form).
Let~{\cite{Serre1977}} $G$ be a compact group with finite dimensional representation $\rho$ acting on the representation space $V = \mathbbm{C}^n$:
\begin{equation}
  \rho : G \rightarrow \tmop{GL} (V), g \mapsto \rho_g,
\end{equation}
where $\tmop{GL} (V)$ is the group the invertible linear maps $V \rightarrow V$, identified with $n \times n$ invertible matrices.
Then $X$ is invariant under $\rho$ if and only if $X = \rho_g X \rho_g^{\dag}$ for all $g \in G$.
In this work, we assume $\rho$ to be unitary\footnote{Representations of compact groups can be made unitary -- thus in the physics literature, nonunitary representations are used primarily for noncompact groups such as the Lorentz group.
  In numerical optimization, forcing representations to be unitary can destroy sparsity or require a field extension (for example, going from the rationals to algebraic numbers). When the representation is nonunitary, note that $X \rightarrow \rho_g X \rho_g^{-1}$ does {\tmem{not}} preserve Hermitianity, rather $X \rightarrow \rho_g X \rho_g^{\dag}$ does, and its invariant subspace is {\em not} the commutant algebra.}. Then $X = \rho_g X \rho_g^{- 1}$, or
\begin{equation}
  [X, \rho_g] = \rho_g X - X \rho_g = 0,
\end{equation}
and we say that $X$ commutes with $\rho$. The set
\begin{equation}
  C_{\rho} = \{ X \in \mathbbm{C}^{n \times n} : [X, \rho_g] = 0, \forall g \in G \}
\end{equation}
is\footnote{The commutant algebra also contains non-Hermitian matrices.} the
{\em commutant}~{\cite[Sec 1.7]{Sagan2001}} or {\em centralizer algebra} of the algebra generated by $\rho$.
The decomposition of $\rho$ into irreducible subrepresentations corresponds to
\begin{equation}
  V = V^1 \oplus \ldots \oplus V^N, \qquad V^i = W^{i, 1} \oplus \ldots \oplus W^{i, M_i},
\end{equation}
where the {\em isotypic components} $V^i$ contain the invariant subspaces $W^{i, j}$ that cannot be decomposed further; in every $V^i$, the subspaces $\{ W^{i, j} \}_j$ correspond to identical subrepresentations.
Explicitly, there exists a unitary change of basis matrix $U$, $U^{\dag} = U^{- 1}$ such that the following is true for all $g \in G$:
\begin{equation}
  \hat{\rho}_g = U \rho_g U^{- 1} = \left(\begin{array}{ccc}
    \hat{\rho}_g^1 &  & \\
    & \ddots & \\
    &  & \hat{\rho}_g^I
  \end{array}\right), \qquad \hat{\rho}_g^i = \left(\begin{array}{ccc}
    \hat{\rho}_g^{i, 1} &  & \\
    & \ddots & \\
    &  & \hat{\rho}_g^{i, M_i}
  \end{array}\right) = \mathbbm{1}_{M_i} \otimes \hat{\rho}_g^{i, 1},
\end{equation}
as $\hat{\rho}^{i, 1}_g = \hat{\rho}_g^{i, 2} = \ldots = \hat{\rho}^{i,M_i}_g$.
We define $D_i = \dim W^{i, j}$; then $D_i$ is the {\em dimension} of the $i$-th irreducible subrepresentation of $\rho$ and $M_i$ its {\em multiplicity} corresponding to the number of copies.
By Schur's lemma~{\cite{Serre1977}}, the matrix $X$ commuting with $\rho$ has a block diagonal form:
\begin{equation}
  \label{Eq:Xhat} \hat{X} = UXU^{\dag} = \left(\begin{array}{ccc}
    \hat{X}^1 &  & \\
    & \ddots & \\
    &  & \hat{X}^I
  \end{array}\right), \qquad \hat{X}^i = \Xi^i \otimes \mathbbm{1}_{D_i},
\end{equation}
where $\Xi^i$ is a $M_i \times M_i$ Hermitian matrix, then $X \succeq 0$ is equivalent to $\{\Xi^i \succeq 0\}_i$.

\begin{definition}
  The {\em key primitive} to reduce the complexity of solving invariant SDPs is defined as follows: given a description of a group $G$, and the explicit map $\rho$ that describes a representation, we ask for the change of basis matrix $U$, the list of dimensions $\{ D_i \}$ and multiplicities $\{ M_i \}$ of the irreducible representations.
\end{definition}

\subsection{Computational requirements}

Let the SDP constraint have dimension $n \times n$ with blocks of size $n_i$ so that $n = n_1 + \ldots + n_I$; as above let $m$ be the number of dual variables.
We consider now the complexity of the widely used interior point primal-dual methods.

For the time complexity, depending on the structure of the problem, Borchers et al.~{\cite{Borchers2007}} observed the following.
When $m \gg n$, the factoring of the Schur complement matrix dominates in $\mathcal{O} (m^3)$.
Otherwise, the Cholesky factorization and eigenvalue computation of the matrices $X$ and $\chi$ usually dominates, in $\mathcal{O} ((n_1)^3 + \ldots (n_I)^3)$.
When the SDP has been formulated in an invariant form, we are usually in this second case.
Storage-wise, the problem data scales in $\mathcal{O} (mn^2)$ in the worst-case when neither sparsity or block structure are present.
The Schur complement matrix requires $\mathcal{O} (m^2)$ storage, and the matrices $X$ and $\chi$ require storage in $\mathcal{O} ((n_1)^2 + \ldots + (n_I)^2)$.
When using our technique, the block diagonalization of a SDP of size $n \times n$ produces a SDP with blocks of size $n'_i = M_i$.

\section{Algorithm}\label{Sec:Algorithms}

We reduce the decomposition of $\rho$ into irreducible subrepresentations into three subtasks.
\begin{enumeratenumeric}
  \item Sample generic elements from $C_{\rho}$.
  
  \item Compute the eigendecomposition of a Hermitian matrix.
  
  \item Find if two irreducible subrepresentations $\sigma^i$ and $\sigma^j$
  of $\rho$ are equivalent, and if yes, find the matrix $A$ such that
  $\sigma^i_g = A \sigma^j_g A^{- 1}$, and apply it.
\end{enumeratenumeric}

\subsection{Sampling from the commutant algebra}

We first obtain a generic element of the commutant algebra.
Compared to other approaches~{\cite{Ibort2016,Maehara2010,Murota2010}}, we sample first a generic Hermitian matrix $X$ from the Gaussian Unitary Ensemble~{\cite{Anderson2009}}.
This ensures that the distribution is invariant under unitary changes of basis, and provides guarantees on eigenvalue separation.
We then project $X$ on the commutant subspace; as $\rho$ is unitary, we have
\begin{equation}
  \begin{array}{cl|lc}
    \text{{\tmem{Finite group}}} &  &  & \text{{\tmem{Compact group}}}\\
    \overline{X} = \frac{1}{| G |} \sum_{g \in G} \rho_g X \rho_g^{- 1} &  & 
    & \overline{X} = \int_G \rho_g X \rho_g^{- 1} \mathd \mu (g)
  \end{array},
\end{equation}
where $\mu$ is the Haar measure of $G$.
We thus reduced the problem of sampling from $C_{\rho}$ to the problem of projecting on $C_{\rho}$.

For a subset $T$ of $G$, we define the partial averaging operator $\Sigma_{\rho, T} [X] = \sum_{g \in T} \rho_g X \rho_g^{- 1} / | T |$.
In our previous work~\cite{Tavakoli2019}, we computed a decomposition of a finite $G$ into a cartesian product of sets $T_1, \ldots, T_{\nu}$ such that every $g \in G$ has a unique decomposition $g = t_{\nu} \ldots t_1$ with $t_i \in T_i$.
Then:
\begin{equation}
  \label{Eq:RelativeReynolds} \overline{X} = \Sigma_{\rho, T_{\nu}}
  [\Sigma_{\rho, T_{\nu - 1}} [\ldots \Sigma_{\rho, T_1} [X]]] .
\end{equation}
For the symmetric group $S_D$ of order $D!$, this reduces the number of actions of $\rho$
from $\mathcal{O} (D!)$ to $\mathcal{O} (D^2)$ with~\footnote{This uses the decomposition of any permutation into a product
of cycles, each of length 2,3, \ldots $D$.} $\nu = D$.

To generalize the approach to compact groups, we work with an oracle that samples elements from the Haar measure $\mu$.
Given an integer $\nu$, we sample small sets\footnote{For finite groups, see~{\cite{Babai1991}} where different sets $\{ T_i \}$ are used in the approximation.} $T_1, \ldots, T_{\nu}$ from the Haar measure $\mu$, typically with $| T_i | = 3$.
We then define $\overline{X}_{\nu} = \Sigma_{\rho, T_{\nu}} [\Sigma_{\rho, T_{\nu - 1}} [\ldots \Sigma_{\rho, T_1} [X]]]$, and observe that $\overline{X}_{\nu}
\rightarrow \overline{X}$ when $\nu \rightarrow \infty$.

We thus reduced the subtask of sampling from $C_{\rho}$ to the task of sampling from $G$ itself.
RepLAB uses the technique~(\ref{Eq:RelativeReynolds}) when dealing with finite groups; for compact groups, we use the iterated averaging on random samples with $\nu = 1000$ in absence of a better termination criterion (work in progress).

\subsection{Computing the eigendecomposition of a Hermitian matrix}

We compute numerically the eigendecomposition of $\overline{X}$, and obtain a change of basis matrix $U$ such that $\hat{X} = U \overline{X} U^{\dag}$ is diagonal.
Up to reordering of eigenvalues, $\hat{X}$ has the form~(\ref{Eq:Xhat}) with fully diagonal blocks $\Xi^i$.
By assumption of genericity, there are no repeated eigenvalues inside each $\Xi^i$ and across them.
Then, we group equal eigenvalues and denote the basis of the corresponding eigenspaces by the matrices $U^1, \ldots, U^M$, so that
\begin{equation}
  \sigma^i : g \mapsto \sigma_g^i = U^i \rho_g (U^i)^{\dag}
\end{equation}
are irreducible subrepresentations of $\rho$.
The computational cost of this step is $\mathcal{O} (n^3)$ with $n$ the dimension of $\rho$.

\subsection{Grouping equivalent representations}

To group the irreducible subrepresentations identified in the previous step by equivalency, we take another sample $\overline{X}'$ from $C_{\rho}$ and apply the following.

\begin{proposition}
  \label{Prop:Equivalent}Given two basis matrices $U^i$ and $U^j$, let
  \begin{equation}
    \sigma^i : g \mapsto \sigma^i_g = U^i \rho_g (U^i)^{\dag}, \qquad \sigma^j
    : g \mapsto \sigma^j_g = U^j \rho_g (U^j)^{\dag},
  \end{equation}
  be two irreducible subrepresentations of $\rho$. Let $\overline{X}'$ be a
  generic sample from $C_{\rho}$, uncorrelated with $U^i, U^j$. Let $F = U^i
  \overline{X}' (U^j)^{\dag}$. If $F = 0$ then (with probability one) $\sigma^i$ and $\sigma^j$ are
  not equivalent. Otherwise, the two subrepresentations are equivalent,
  $\sigma^i_g = F \sigma_g^j F^{- 1}$, and $\alpha F$ is unitary for some $\alpha \in \mathbbm{C}$ .
\end{proposition}

\begin{proof}
  Noting that $(U^i)^{\dag} U^i$ is a projector on the corresponding invariant subspace, we verify that $F$ is an equivariant map: $\sigma^i_g F = F \sigma^j_g$ for all $g \in G$.
  We now use Schur's lemma~{\cite[Prop. 4]{Serre1977}}.
  By the assumption of genericity, $F = 0$ happens only when $F$ has to be zero, and $\sigma^i$ is inequivalent to $\sigma^j$.
  Otherwise, there is a unitary~{\cite{Mozrzymas2014}} matrix $A$ such that $\sigma^i_g = A \sigma^j_g A^{- 1}$, and thus $\sigma_g^j  (A^{- 1} F) = (A^{- 1} F) \sigma_g^j$.
  By Schur's lemma $A^{- 1} F = \alpha \mathbbm{1}$, and we have necessarily $F = \alpha A$ with $\alpha \in \mathbbm{C}$.
\end{proof}

Checking every pair $(i, j)$ with $i > j$, we group the bases $U^i$ into isotypic components, and make sure that equivalent irreducible
subrepresentations are all expressed in the same basis.
From this grouping, we compute the dimensions $\{ D_i \}$ and corresponding multiplicities $\{ M_i \}$.

The ideas presented above can be adapted to provide new primitives in the software.
Mozrzymas et al.~\cite{Mozrzymas2014} considered the problem of deciding whether two arbitrary irreducible representations of $G$, $\sigma^1: G \to \tmop{GL}(\mathbbm{C}^{n_1})$ and $\sigma^2: G \to \tmop{GL}(\mathbb{C}^{n_2})$ are equivalent, and computing the change of basis matrix.
We apply Proposition~\ref{Prop:Equivalent}, constructing a new representation $\rho: G \to \tmop{GL}(\mathbbm{C}^{n_1+n_2})$ such that $\rho_g = \sigma^1_g \oplus \sigma^2_g$, providing an algorithm that generalizes to compact groups.
We leave as an open question the construction of other primitives.

\subsection{Usage in practice}

The algorithm above only require a way to sample from the group and the representation image function.
This is appropriate to decompose, for example, tensor products of the defining representation of the unitary group: RepLAB provides a way to sample from the unitary group, which is confounded with the matrix $u$ of the defining representation, and then the image function is simply $u \mapsto u \otimes \ldots \otimes u$.
For representations of finite groups, the user can define a permutation group using its generators, and then the representation by the images of those generators.
No additional information is required.

\section{Other features of RepLAB}\label{Sec:Features}

RepLAB supports the decomposition of representations over both $\mathbbm{R}$ and $\mathbbm{C}$.
The best solvers available at the date of this writing do not support optimization over Hermitian matrices, in which case a wasteful scheme is used~{\cite{Gilbert2017a}}, losing some of the gains of symmetry reduction.
Real commutant algebras have a more complex structure; they includes three types of irreducible representations~\cite{Maehara2010}.
We found efficient and simple methods to address this challenge\footnote{In contrast, Maehara et al.~\cite{Maehara2010} prescribe the use of a few exotic matrix decomposition techniques, some of which do not have open source implementations available.}.

For basic objects such as permutations and tuples, RepLAB reuses primitive MATLAB types; for example permutations are represented as row vectors of images.
For further accessibility, RepLAB is compatible with the open source clone GNU Octave~{\cite{Rothlisberger2002}}, and that compatibility is tested at each release.

RepLAB contains heuristics to preserve sparsity, and recover rational basis matrices in some cases.
It also contains out-of-the-box support for permutation groups, groups of signed permutations relevant for the study of correlation Bell inequalities, and the unitary group; those groups can be combined using standard group constructions such as direct products,
semidirect products and wreath products relevant for quantum information scenarios~{\cite{Renou2017}}; and representations of the factor groups can be combined to create representations of the product.

RepLAB extends the toolbox YALMIP~{\cite{Lofberg2004}} by providing symmetry-satisfying SDP variables.
It can also impose symmetry on existing SDP variables, easing the construction of invariant optimization problems.
It can block diagonalize SDP data provided in an extension of the SeDuMi MAT format~{\cite{Sturm1999}}.

\section{Open questions}\label{Sec:Open}

A main open question is to extend the applicability of RepLAB beyond block-diagonalizing semidefinite programs: what other operations can be based on a sampling oracle?
Other questions or challenges are discussed below. 

\subsection{Scaling}

How do the algorithms employed in the three subtasks of Section~\ref{Sec:Algorithms} scale?
Can we do better?
The first subtask requires $\mathcal{O} (n^2)$ storage and its time complexity is currently unknown.
The second subtask has currently time complexity $\mathcal{O}(n^3)$ and a running time similar to the eigendecomposition step in a single solver iteration working on the original SDP.
By replacing the standard eigendecomposition algorithm by the Lanczos algorithm~(see {\cite{Maslen2003}}), we could reduce the time complexity to $\mathcal{O} (Nn^2 + N^2)$, where $N$ is the number of irreducible subrepresentations.
Are better asymptotics available for representations of compact groups?

Another route is to use the black-box algorithms above in the last resort, and
exploit partial information about the structure of the groups/representations.
As discussed above, RepLAB has built in support for standard textbook group and representation constructions, but is not exploiting that
structure at the moment.
Nevertheless, RepLAB has hooks in parts of its internals so that specialized methods can be registered (see~{\cite{Breuer1998}} for a similar system).

\subsection{Precision}

Currently, RepLAB works in hardware double floating point precision.
RepLAB contains several magic $\varepsilon$ thresholds appropriate for that precision and problems with $n \lessapprox 10000$.
We are currently working on analyzing the numerical robustness of our algorithm and will replace them by proper error bounds (in the same spirit as the bounds in~{\cite{Babai1991}}).
Then in principle, RepLAB could provide approximate subrepresentations bases with arbitrary precision.

The current version of RepLAB performs limited exact solution recovery in the case of rational representations~{\cite{Plesken1996}}.
It is known that irreducible representations of the symmetric group can always be realized over $\mathbbm{Q}$; we observed numerically that tensor products of representations of the unitary group, possibly partially transposed, decompose with change of basis in $\mathbbm{Q}$.
Is there a general principle at play there?
For finite groups, it is known that small extensions of the rationals are necessary, with complexity proportional to the exponent of the group.
When and how can exact results be recovered?

Even assuming that only approximative solutions are available, another possibility is to incorporate this source of error in the overall error analysis.
After all, solvers return slightly infeasible solutions: the primal-dual gap cannot be blindly trusted.
By considering the pair RepLAB-solver as a single black box, we can use standard certification techniques such as verified semidefinite programming~{\cite{Jansson2006}}.

\section{Conclusion}

We presented a toolbox to numerically decompose arbitrary finite dimensional representation of compact groups.
Surprisingly, the user needs to provide only little information about the group and its representation, and there is no need to compute much structure to accomplish our task.
We foresee RepLAB having impact in two ways.
The first one is reducing the computational cost of solving SDPs, thus expanding the applicability of a wide variety of quantum information methods.
The second one is pedagogical: by delegating all computations to the software, a hands-on approach to representation theory can be taught, focusing on the physics by working on concrete examples right at the start.
Along the same line, RepLAB can be used to quickly check whether an algebraic analysis of the symmetries of a problem is worthwhile: while the bases returned are approximate, the dimensions and multiplicities of the irreducible subrepresentations are themselves not approximate.

\paragraph*{Acknowledgments. ---}We acknowledge useful discussions with David Gross, Elie Wolfe and Markus Heinrich.
This research was supported by Perimeter Institute for Theoretical Physics.
Research at Perimeter Institute is supported in part by the Government of Canada through the Department of Innovation, Science and Economic Development Canada and by the Province of Ontario through the Ministry of Economic Development, Job Creation and Trade.
This publication was made possible through the support of a grant from the John Templeton Foundation.
The opinions expressed in this publication are those of the authors and do not necessarily reflect the views of the John Templeton Foundation.
FMM was funded by the DFG project number 4334.

\bibliographystyle{spphys}
\bibliography{rosset}

\begin{thebibliography}{10}
\providecommand{\url}[1]{{#1}}
\providecommand{\urlprefix}{URL }
\expandafter\ifx\csname urlstyle\endcsname\relax
  \providecommand{\doi}[1]{DOI \discretionary{}{}{}#1}\else
  \providecommand{\doi}{DOI \discretionary{}{}{}\begingroup
  \urlstyle{rm}\Url}\fi

\bibitem{Szanton1992}
A.~Szanton, in \emph{The {{Recollections}} of {{Eugene P}}. {{Wigner}}}, ed. by
  A.~Szanton ({Springer US}, {Boston, MA}, 1992), pp. 115--125.
\newblock \doi{10.1007/978-1-4899-6313-0_8}

\bibitem{Rosset2014a}
D.~Rosset, J.D. Bancal, N.~Gisin, J. Phys. A: Math. Theor. \textbf{47}(42),
  424022 (2014).
\newblock \doi{10.1088/1751-8113/47/42/424022}

\bibitem{Brunner2014}
N.~Brunner, D.~Cavalcanti, S.~Pironio, V.~Scarani, S.~Wehner, Rev. Mod. Phys.
  \textbf{86}(2), 419 (2014).
\newblock \doi{10.1103/RevModPhys.86.419}

\bibitem{Sliwa2003}
C.~{\'S}liwa, Physics Letters A \textbf{317}(3-4), 165 (2003).
\newblock \doi{10.1016/S0375-9601(03)01115-0}

\bibitem{Collins2002}
D.~Collins, N.~Gisin, N.~Linden, S.~Massar, S.~Popescu, Phys. Rev. Lett.
  \textbf{88}(4), 040404 (2002).
\newblock \doi{10.1103/PhysRevLett.88.040404}

\bibitem{Mozrzymas2014b}
M.~Mozrzymas, M.~Horodecki, M.~Studzi{\'n}ski, Journal of Mathematical Physics
  \textbf{55}(3), 032202 (2014).
\newblock \doi{10.1063/1.4869027}

\bibitem{Studzinski2013}
M.~Studzi{\'n}ski, M.~Horodecki, M.~Mozrzymas, J. Phys. A: Math. Theor.
  \textbf{46}(39), 395303 (2013).
\newblock \doi{10.1088/1751-8113/46/39/395303}

\bibitem{Mozrzymas2018}
M.~Mozrzymas, M.~Studzi{\'n}ski, M.~Horodecki, J. Phys. A: Math. Theor.
  \textbf{51}(12), 125202 (2018).
\newblock \doi{10.1088/1751-8121/aaad15}

\bibitem{Mozrzymas2014a}
M.~Mozrzymas, M.~Horodecki, M.~Studzi{\'n}ski, Journal of Mathematical Physics
  \textbf{55}(3), 032202 (2014).
\newblock \doi{10.1063/1.4869027}

\bibitem{Studzinski2014}
M.~Studzi{\'n}ski, P.~{\'C}wikli{\'n}ski, M.~Horodecki, M.~Mozrzymas, Phys.
  Rev. A \textbf{89}(5), 052322 (2014).
\newblock \doi{10.1103/PhysRevA.89.052322}

\bibitem{Renou2017}
M.O. Renou, D.~Rosset, A.~Martin, N.~Gisin, Journal of Physics A: Mathematical
  and Theoretical \textbf{50}(25), 255301 (2017)

\bibitem{Rosset2019losr}
D.~Rosset, D.~Schmid, F.~Buscemi,   (in preparation)

\bibitem{Choi1975}
M.D. Choi, Linear Algebra and its Applications \textbf{10}(3), 285 (1975).
\newblock \doi{10.1016/0024-3795(75)90075-0}

\bibitem{Jamiolkowski1972}
A.~Jamio{\l}kowski, Reports on Mathematical Physics \textbf{3}(4), 275 (1972).
\newblock \doi{10.1016/0034-4877(72)90011-0}

\bibitem{Doherty2002}
A.C. Doherty, P.A. Parrilo, F.M. Spedalieri, Phys. Rev. Lett. \textbf{88}(18),
  187904 (2002).
\newblock \doi{10.1103/PhysRevLett.88.187904}

\bibitem{Doherty2004}
A.C. Doherty, P.A. Parrilo, F.M. Spedalieri, Phys. Rev. A \textbf{69}(2),
  022308 (2004).
\newblock \doi{10.1103/PhysRevA.69.022308}

\bibitem{Doherty2005}
A.C. Doherty, P.A. Parrilo, F.M. Spedalieri, Phys. Rev. A \textbf{71}(3),
  032333 (2005).
\newblock \doi{10.1103/PhysRevA.71.032333}

\bibitem{Navascues2009}
M.~Navascu{\'e}s, M.~Owari, M.B. Plenio, Phys. Rev. Lett. \textbf{103}(16),
  160404 (2009).
\newblock \doi{10.1103/PhysRevLett.103.160404}

\bibitem{Navascues2008a}
M.~Navascu{\'e}s, S.~Pironio, A.~Ac{\'i}n, New J. Phys. \textbf{10}(7), 073013
  (2008).
\newblock \doi{10.1088/1367-2630/10/7/073013}

\bibitem{Navascues2007}
M.~Navascu{\'e}s, S.~Pironio, A.~Ac{\'i}n, Phys. Rev. Lett. \textbf{98}(1),
  010401 (2007).
\newblock \doi{10.1103/PhysRevLett.98.010401}

\bibitem{Navascues2012}
M.~Navascu{\'e}s, S.~Pironio, A.~Ac{\'i}n, in \emph{Handbook on
  {{Semidefinite}}, {{Conic}} and {{Polynomial Optimization}}}, ed. by M.F.
  Anjos, J.B. Lasserre, no. 166 in International {{Series}} in {{Operations
  Research}} \& {{Management Science}} ({Springer US}, 2012), pp. 601--634.
\newblock \doi{10.1007/978-1-4614-0769-0_21}

\bibitem{Navascues2015a}
M.~Navascu{\'e}s, A.~Feix, M.~Ara{\'u}jo, T.~V{\'e}rtesi, Phys. Rev. A
  \textbf{92}(4), 042117 (2015).
\newblock \doi{10.1103/PhysRevA.92.042117}

\bibitem{Wolfe2019}
E.~Wolfe, A.~{Pozas-Kerstjens}, M.~Grinberg, D.~Rosset, A.~Ac{\'i}n,
  M.~Navascues, arXiv:1909.10519 [quant-ph]  (2019)

\bibitem{Poland2019}
D.~Poland, S.~Rychkov, A.~Vichi, Rev. Mod. Phys. \textbf{91}(1), 015002 (2019).
\newblock \doi{10.1103/RevModPhys.91.015002}

\bibitem{Go2019}
M.~Go, Y.~Tachikawa, J. High Energ. Phys. \textbf{2019}(6), 84 (2019).
\newblock \doi{10.1007/JHEP06(2019)084}

\bibitem{Schneider1990}
G.J.A. Schneider, Journal of Symbolic Computation \textbf{9}(5), 601 (1990).
\newblock \doi{10.1016/S0747-7171(08)80077-6}

\bibitem{Conlon1990}
S.B. Conlon, Journal of Symbolic Computation \textbf{9}(5), 535 (1990).
\newblock \doi{10.1016/S0747-7171(08)80072-7}

\bibitem{Baum1994}
U.~Baum, M.~Clausen, Mathematics of Computation \textbf{63}(207), 351 (1994).
\newblock \doi{10.2307/2153580}

\bibitem{Gap}
\emph{{{GAP}} -- {{Groups}}, {{Algorithms}}, and {{Programming}}, {{Version}}
  4.7.8} ({The GAP Group}, 2015)

\bibitem{Tavakoli2019}
A.~Tavakoli, D.~Rosset, M.O. Renou, Phys. Rev. Lett. \textbf{122}(7), 070501
  (2019).
\newblock \doi{10.1103/PhysRevLett.122.070501}

\bibitem{Gilbert2017a}
J.C. Gilbert, C.~Josz, Plea for a semidefinite optimization solver in complex
  numbers \textendash{} {{The}} full report.
\newblock Research report, {INRIA Paris, LAAS} (2017)

\bibitem{Sturm2002}
J.F. Sturm, Optimization Methods and Software \textbf{17}(6), 1105 (2002).
\newblock \doi{10.1080/1055678021000045123}

\bibitem{Rosset2018}
D.~Rosset, arXiv:1808.09598 [quant-ph]  (2018)

\bibitem{Gatermann2004}
K.~Gatermann, P.A. Parrilo, Journal of Pure and Applied Algebra
  \textbf{192}(1\textendash{}3), 95 (2004).
\newblock \doi{10.1016/j.jpaa.2003.12.011}

\bibitem{Serre1977}
J.P. Serre, \emph{Linear {{Representations}} of {{Finite Groups}}}.
\newblock Graduate Texts in {{Mathematics}} ({Springer}, 1977)

\bibitem{Sagan2001}
B.~Sagan, \emph{The {{Symmetric Group}}: {{Representations}}, {{Combinatorial
  Algorithms}}, and {{Symmetric Functions}}}, 2nd edn.
\newblock Graduate {{Texts}} in {{Mathematics}} ({Springer-Verlag}, {New York},
  2001).
\newblock \doi{10.1007/978-1-4757-6804-6}

\bibitem{Borchers2007}
B.~Borchers, J.G. Young, Comput Optim Appl \textbf{37}(3), 355 (2007).
\newblock \doi{10.1007/s10589-007-9030-3}

\bibitem{Ibort2016}
A.~Ibort, A.~{L{\'o}pez-Yela}, J.~Moro, arXiv:1610.01054 [math-ph]  (2016)

\bibitem{Maehara2010}
T.~Maehara, K.~Murota, Japan J. Indust. Appl. Math. \textbf{27}(2), 263 (2010).
\newblock \doi{10.1007/s13160-010-0007-8}

\bibitem{Murota2010}
K.~Murota, Y.~Kanno, M.~Kojima, S.~Kojima, Japan J. Indust. Appl. Math.
  \textbf{27}(1), 125 (2010).
\newblock \doi{10.1007/s13160-010-0006-9}

\bibitem{Anderson2009}
G.W. Anderson, A.~Guionnet, O.~Zeitouni, \emph{An {{Introduction}} to {{Random
  Matrices}}}, 1st edn. ({Cambridge University Press}, {New York}, 2009)

\bibitem{Babai1991}
L.~Babai, K.~Friedl, in \emph{[1991] {{Proceedings}} 32nd {{Annual Symposium}}
  of {{Foundations}} of {{Computer Science}}} (1991), pp. 733--742.
\newblock \doi{10.1109/SFCS.1991.185442}

\bibitem{Mozrzymas2014}
M.~Mozrzymas, M.~Studzi{\'n}ski, M.~Horodecki, J. Phys. A: Math. Theor.
  \textbf{47}(50), 505203 (2014).
\newblock \doi{10.1088/1751-8113/47/50/505203}

\bibitem{Rothlisberger2002}
B.~R{\"o}thlisberger, J.~Lehmann, D.~Loss, Computer Physics Communications
  \textbf{183}(1), 155 (2002).
\newblock \doi{10.1016/j.cpc.2011.08.012}

\bibitem{Lofberg2004}
J.~Lofberg, in \emph{2004 {{IEEE International Conference}} on {{Robotics}} and
  {{Automation}} ({{IEEE Cat}}. {{No}}.{{04CH37508}})} (2004), pp. 284--289.
\newblock \doi{10.1109/CACSD.2004.1393890}

\bibitem{Sturm1999}
J.F. Sturm, Optimization Methods and Software \textbf{11}(1-4), 625 (1999).
\newblock \doi{10.1080/10556789908805766}

\bibitem{Maslen2003}
D.K. Maslen, M.E. Orrison, D.N. Rockmore, SIAM Journal on Matrix Analysis and
  Applications \textbf{25}(3), 784 (2003)

\bibitem{Breuer1998}
T.~Breuer, S.~Linton, in \emph{Proceedings of the 1998 {{International
  Symposium}} on {{Symbolic}} and {{Algebraic Computation}}} ({ACM}, {New York,
  NY, USA}, 1998), {{ISSAC}} '98, pp. 38--45.
\newblock \doi{10.1145/281508.281540}

\bibitem{Plesken1996}
W.~Plesken, B.~Souvignier, Experimental Mathematics \textbf{5}(1), 39 (1996).
\newblock \doi{10.1080/10586458.1996.10504337}

\bibitem{Jansson2006}
C.~Jansson, DELTA \textbf{1}, 4 (2006)

\end{thebibliography}

\end{document}